\documentclass[a4paper,11pt]{article}

\usepackage{authblk}

\usepackage{amsmath}
\usepackage{amsthm}
\usepackage{amssymb}

\usepackage{algorithmic}
\usepackage{algorithm}

\usepackage{graphicx}
\usepackage{epsfig}

\usepackage[left=1.0in,top=1.0in,right=1.0in,bottom=1.0in,nohead]{geometry}

\theoremstyle{definition}
\newtheorem{definition}{Definition}
\newtheorem{observation}{Observation}

\newtheorem{theorem}{Theorem}
\newtheorem{corollary}[theorem]{Corollary}
\newtheorem{lemma}[theorem]{Lemma}

\newcommand{\G}{\mathcal G}
\newcommand{\V}{\mathcal V}
\newcommand{\E}{\mathcal E}

\long\def\longdelete#1{}

\title{\bf Approximation Algorithms for the Capacitated Domination Problem\thanks{
This work was supported in part by the National Science Council, Taipei 10622, Taiwan, under
the Grants NSC98-2221-E-001-007-MY3 and NSC98-2221-E-001-008-MY3.}}

\author[ ]{Mong-Jen Kao}
\author[ ]{Han-Lin Chen}

\affil[ ]{Department of Computer Science and Information Engineering, \newline
National Taiwan University, Taiwan.}

\affil[ ]{\newline Emails: d97021@csie.ntu.edu.tw, kalent37@ms89.url.com.tw}

\date{\empty}

\begin{document}

\maketitle


\begin{abstract}
We consider the {\em Capacitated Domination} problem, which models a
service-requirement assignment scenario and is also a
generalization of the well-known {\em Dominating Set} problem. In
this problem, given a graph with three parameters defined on each vertex, namely cost,
capacity, and demand, we want to find
an assignment of demands to vertices of least cost such that the demand of each vertex is
satisfied subject to the capacity constraint of each vertex providing the service.

In terms of polynomial time approximations, we present logarithmic approximation algorithms with respect to different demand assignment models for this problem on general
graphs, which also establishes the corresponding approximation
results to the well-known approximations of the
traditional {\em Dominating Set} problem. Together with our
previous work, this closes the problem of generally approximating
the optimal solution. On the other hand, from the perspective of
parameterization, we prove that this problem is {\it W[1]}-hard when parameterized by a structure of the graph called treewidth. Based on this hardness result, we present exact fixed-parameter tractable algorithms when parameterized by treewidth and maximum capacity of the vertices.
This algorithm is further extended to obtain pseudo-polynomial time approximation schemes for planar graphs.
\end{abstract}


\section{Introduction}

For decades, {\em Dominating Set} problem has been one of the most
fundamental and well-known problems in both 
graph theory and combinatorial optimization. Given a graph $G=(V,E)$ and an integer $k$,
{\em Dominating Set} asks for a subset $D\subseteq V$ whose cardinality
does not exceed $k$ such that every vertex in the graph either
belongs to this set or has a neighbor which does. As this problem is
known to be NP-hard, approximation algorithms have been proposed in
the literature. On one hand, a simple greedy algorithm is shown to
achieve a guaranteed ratio of $O(\ln n)$ \cite{VC79,804034,LL75}, where $n$ is the
number of vertices, which is later proven to be the approximation
threshold by Feige \cite{285059}. On the other hand, algorithms
based on dual-fitting provide a guaranteed ratio of $\Delta$
\cite{DSH82}, where $\Delta$ is the maximum degree of the vertices of the graph.

\smallskip

In addition to polynomial time approximations, {\em Dominating
Set} has its special place from the perspective of parameterized complexity as well \cite{RGDMRF99,JFMG06,RN06}. In
contrast to {\em Vertex Cover}, which is fixed-parameter tractable (FPT),
{\em Dominating Set} has been proven to be {\it W[2]}-complete when
parameterized by solution size, in the sense that no fixed-parameter
algorithm exists (with respect to solution size) unless FPT={\it W[2]}. Though {\em Dominating Set} is a fundamentally hard problem in the parameterized $W$-hierarchy, it has been used as a benchmark problem for developing {\it sub-exponential time} parameterized algorithms \cite{JAHLBHFTKRN02,1101823,1139978} and 
{\it linear size kernels} have been obtained in planar graphs \cite{990309,JFMG06,DBLP:conf/icalp/GuoN07,RN06}, and more generally, in graphs that exclude a fixed graph $H$ as a minor.

\smallskip

Besides {\em Dominating Set} problem itself, a vast body of
work has been proposed in the literature, considering possible variations from purely theoretical aspects to practical applications. See \cite{TWHSHPS98,FSR78} for a detailed survey. In particular, variations of {\em Dominating Set} problem occur in numerous practical settings, ranging from strategic decisions, such as locating radar stations or emergency services, to computational biology and to voting systems. For example, Haynes et al. \cite{587927} considered {\em Power Domination Problem} in electric networks \cite{587927,DBLP:conf/cocoon/LiaoL05} while Wan et al. \cite{PJWKMAOF03} considered {\em Connected Domination Problem} in wireless ad hoc networks.

\smallskip

Motivated by a general service-requirement assignment model, Kao et al., \cite{MJKCSLDTL09} considered a generalized domination problem called {\em Capacitated Domination}. In this problem, the input graph is given with tri-weighted vertices, referred to as {\em cost}, {\em capacity}, and {\em demand}, respectively. The demand of a vertex stands for the amount of service it requires from its adjacent vertices (including itself) while the capacity of a vertex represents the amount of service it can provide when it's selected as a server. The goal of this problem is to find a dominating multi-set as well as a demand assignment function such that the overall cost of the multi-set is minimized. For different underlying applications, there are two different demand assignment models, namely {\em splittable} demand model and {\em unsplittable} demand model, depending on whether or not the demand of a vertex is allowed to be served by different vertices. Moreover, there has been work studying the variation when the number of copies, or {\it multiplicity} of each vertex in the dominating multi-set, is limited, referred to as {\em hard} capacity, and as {\em soft} capacity when no such limit is specified. Kao et al., \cite{MJKCSLDTL09} considered the soft capacitated domination problem with splittable demand and provided a $(\Delta+1)$-approximation for general graphs, where $\Delta$ is the degree of the graph. For special graph classes, they proved that even when the input graph is restricted to a tree, the soft capacitated domination problem with splittable demand remains NP-hard, for which they also presented a polynomial time approximation scheme. Dom et al., \cite{DBLP:conf/iwpec/DomLSV08} considered the hard capacitated domination problem with uniform demand and showed that this problem is {\it W[1]}-hard even when parameterized by treewidth and solution size.

\smallskip


In this paper, we consider the (soft) {\em Capacitated Domination} problem and present logarithmic approximation algorithms with respect to different demand assignment models on general graphs. Specifically, we provide a $(\ln n)$-approximation
for weighted unsplittable demand model, a $(4\ln n+2)$-approximation for weighted splittable demand model, and a $(2\ln n+1)$-approximation for unweighted splittable demand model, where $n$ is the number of vertices. Together with the $(\Delta+1)$-approximation result given by Kao et al., \cite{MJKCSLDTL09}, this establishes a corresponding near-optimal approximation result to the original {\em Dominating Set} problem. Although the result may look natural, the greedy choice we make is not obvious when non-uniform capacity as well as non-uniform demand is taken into consideration. On the other hand, from the perspective of parameterization, we prove that this problem is {\it W[1]}-hard when parameterized by a structure of the graph, called the treewidth, and present exact FPT algorithms when parameterized by the treewidth and the maximum capacity of the vertices. 
This algorithm is further extended to obtain pseudo-polynomial time approximation schemes for planar graphs, based on a framework due to Baker \cite{174650}.

\smallskip

The rest of this paper is organized as follows. In Section~\ref{preliminary}, we give formal definitions and notation adopted in the paper. In Section~\ref{logarithmic_approximations}, we present our ideas and algorithms that achieve the aforementioned approximation guarantees. We present the parameterized results in Section~\ref{parameterized_results} and conclude by listing some future work in Section~\ref{conclusion}.


\section{Preliminary} \label{preliminary}

We assume that all the graphs considered in this paper are simple and undirected. Let $G=(V,E)$ be a graph with vertex set $V$ and edge set $E$. A vertex $v \in V$ is said to be adjacent to a vertex $u \in V$ if $(u,v) \in E$. The set of neighbors of a vertex $v \in V$ is denoted by $N_G(v) = \{u:(u,v) \in E\}$. The closed neighborhood of $v \in V$ is denoted by $N_G[v] = N_G(v) \cup \{v\}$.  The subscript $G$ in $N_G[v]$ will be omitted when there is no confusion.

\smallskip

Consider a graph $G=(V,E)$ with tri-weighted vertices, referred to as the cost, the capacity, and the demand of each vertex $u \in V$, denoted by $w(u)$, $c(u)$, and $d(u)$, respectively. Let $D$ denote a multi-set of vertices of $V$ and for any vertex $u \in V$, let $x_D(u)$ denote the {\it multiplicity} of $u$ or the number of times of $u$ in $D$.
The cost of $D$, denoted $w(D)$, is defined to be $w(D)= \sum_{u\in D} w(u)\cdot x_D(u)$.

\begin{definition}[Capacitated Dominating Set]
A vertex multi-subset $D$ is said to be a feasible capacitated dominating set with respect to a demand assignment function $f$ if the following conditions hold.
\begin{itemize}
    \item {\bf Demand constraint:} $\sum_{u \in N_G[v]}f(v,u) \ge d(v)$, for each $v \in V$.
    \item {\bf Capacity constraint:} $\sum_{u \in N_G[v]}f(u,v) \le c(v)\cdot x_D(v)$, for each $v \in V$.
\end{itemize}
\end{definition}

Given a problem instance, the capacitated domination problem asks for a capacitated dominating multi-set $D$ and demand assignment function $f$ such that $w(D)$ is minimized. For unsplittable demand model we require that $f(u,v)$ is either $0$ or $d(u)$ for each edge $(u,v) \in E$. Note that since it is already NP-hard\footnote{This can be verified by making a reduction from {\sc Subset Sum}.} to compute a feasible demand assignment function from a given feasible capacitated dominating multi-set when the demand cannot be split, it is natural to require the demand assignment function be specified, in addition to the optimal vertex multi-set itself.

\smallskip

Parameterized complexity is a well-developed framework for studying the computationally hard problem \cite{RGDMRF99,JFMG06,RN06}. A problem is called {\it fixed-parameter tractable} (FPT) with respect to a parameter $k$ if it can be solved in time $f(k)\cdot n^{O(1)}$, where $f$ is a computable function depending only on $k$. Problems (along with its defining parameters) belonging to {\it W[t]}-hard for any $t \ge 1$ are believed not to admit any FPT algorithms (with respect to the specified parameters). Now we define the notion of parameterized reduction.

\begin{definition}
Let $A$ and $B$ be two parameterized problems. We say that $A$ reduces to $B$ by a standard parameterized reduction if there exists an algorithm $\Phi$ that transforms $(x,k)$ into $(x^\prime, g(k))$ in time $f(k)\cdot \left|x\right|^\alpha$, where $f,g:\mathcal{N} \rightarrow \mathcal{N}$ are arbitrary functions and $\alpha$ is a constant independent of $\left|x\right|$ and $k$, such that $(x,k)\in A$ if and only if $(x^\prime,g(k)) \in B$.
\end{definition}

Next we define the concept of {\em tree decomposition} \cite{HLBAMCAK08,DBLP:books/sp/Kloks94}.

\begin{definition}[Tree Decomposition of a Graph]
A tree decomposition of a graph $G=(V,E)$ is a pair $(X=\left\{X_i:i\in I\right\}, T=(I,F))$ where each node $i\in I$ has associated with it a subset of vertices $X_i \subseteq V$, called the bag of $i$, such that
\begin{enumerate}
    \item{Each vertex belongs to at least one bag: $\bigcup_{i \in I}X_i = V$.}
    \item{For all edges, there is a bag containing both its end-points.}
    \item{For all vertices $v \in V$, the set of nodes $\{i \in I:v \in X_i\}$ induces a subtree of $T$.}
\end{enumerate}
\end{definition}

The width of a tree decomposition is $\max_{i \in I}\left|X_i\right|$. The treewidth of a graph $G$ is the minimum width over all tree decompositions of $G$.

\begin{definition}[Nice Tree Decomposition \cite{DBLP:books/sp/Kloks94}]
A tree decomposition $(X,T)$ is a nice tree decomposition if one can root $T$ in such a way that each node $i \in I$ is of one of the four following types.
\begin{enumerate}
    \item{{\em Leaf}: node $i$ is a leaf of $T$, and $\left|X_i\right| = 1$.}
    \item{{\em Join}: node $i$ has exactly two children, say $j_1$ and $j_2$, and $X_i = X_{j_1} = X_{j_2}$.}
    \item{{\em Introduce}: node $i$ has exactly one child, say $j$, and there is a vertex $v \in V$ such that $X_i = X_j \cup \left\{v\right\}$.}
    \item{{\em Forget}: node $i$ has exactly one child, say $j$, and there is a vertex $v \in V$ such that $X_j = X_i \cup \left\{v\right\}$.}
\end{enumerate}
\end{definition}

Given a tree decomposition of width $k$, a nice tree decomposition of the same width can be found in linear time \cite{DBLP:books/sp/Kloks94}.


\section{Logarithmic Approximation} \label{logarithmic_approximations}

In this section, we present logarithmic approximation algorithms for
capacitated domination problems with respect to different cost and
demand models. Specifically, we provide a $(\ln n)$-approximation
for weighted unsplittable demand model, a $(4\ln n+2)$-approximation for weighted splittable demand model, and a
$(2\ln n+1)$-approximation for unweighted splittable demand model,
where $n$ is the number of vertices.

The main idea is based on greedy approach in the sense that we keep choosing
a vertex with the best efficiency in each iteration until the whole
graph is dominated.
By best efficiency we mean the maximum cost-efficiency ratio defined for 
each vertex in the remaining graph.
We describe the results in more detail in the following
subsections.

\subsection{Weighted Unsplittable Demand}
\label{section_weighted_unsplittable}

In this section, we consider the {\em weighted capacitated
domination problem with unsplittable demand} and provide a simple
greedy algorithm that achieves the approximation guarantee of $\ln
n$.

\smallskip

Let $U$ be the set of vertices which are not dominated yet.
Initially, we have $U = V$. For each vertex $u \in V$, let
$N_{ud}[u] = U \cap N[u]$ be the set of undominated vertices in the closed neighborhood of $u$.
Without loss of generality, we shall assume that
the elements of $N_{ud}[u]$, denoted by $v_{u,1}, v_{u,2}, \ldots, v_{u,{|N_{ud}[u]|}}$, are
sorted in non-decreasing order of their demands in the remaining
section.

\smallskip

In each iteration, the algorithm chooses a vertex of the most efficiency
from $V$, where the efficiency of a vertex, say $u$, is defined by
the largest effective-cost ratio of the number of vertices dominated by $u$ over the total cost. That is, 
$$\max_{1\le i\le
|N_{ud}[u]|}\frac{i}{w(u)\cdot x_u(i)},$$ where $$x_u(i) = \left\lceil\frac{\sum_{1\le j\le i}d(v_{u,j})}{c(u)}\right\rceil$$ is
the number of copies of $u$ selected in order to dominate
$v_{u,1}, v_{u,2}, \ldots,$ and $v_{u,i}$. A high-level description of this algorithm
is presented in Figure~\ref{Algorithm for Weighted Unsplittable Demand}.

\begin{figure*}[t]
\rule{\linewidth}{0.2mm}
\medskip
{{\sc Algorithm} {\em Unsplit-Log-Approx
}}

\begin{algorithmic}[1]
\STATE $U \longleftarrow V$
\WHILE{$U \ne \phi$}
    \STATE Pick a vertex in $V$ with the most efficiency, say $u$.
    \STATE let $k = \arg \max_{1\le i\le |N_{ud}[u]|}\frac{i}{w(u)\cdot x_u(i)}$.
    \STATE Assign the demand of each vertex in $\{v_{u,1}, v_{u,2}, \ldots, v_{u,k}\}$ to $u$ and remove them from $U$.
\ENDWHILE
\STATE compute from the assignment the weight of the dominating set, and return the result.

\end{algorithmic}
\rule{\linewidth}{0.2mm} 
\caption{The pseudo-code for the 
weighted unsplittable demand model.} \label{Algorithm for
Weighted Unsplittable Demand}
\end{figure*}

In iteration $j$, let $OPT_j$ be the cost of the optimal solution
for the remaining problem instance, which is clearly upper bounded
by the cost, $OPT$, of the optimal solution for the input instance.
Let the number of undominated vertices at the beginning of iteration $j$ be
$n_j$, and the number of vertices that are newly dominated in iteration $j$ be $k_j$.

\smallskip

Denote by $S_j$ the cost in iteration $j$. Note that $S_j = w(u) \cdot x_u(k_j)$, where $u$ is the most efficient vertex chosen
in iteration $j$. Assume that the algorithm repeats for $m$
iterations. 
We have the following lemma.

\begin{lemma} \label{lemma_greedy_weighted_unsplittable}
For each $j$, $1\le j \le m$, we have $S_j\le \frac{k_j}{n_j}\cdot OPT_j$.
\end{lemma}

\begin{proof}
Since we always choose the vertex with the maximum 
efficiency, the efficiency is no less than that of each vertex
chosen in $OPT_j$, which is no less than the average of $OPT_j$. Therefore we have
$\frac{k_j}{S_j} \ge \frac{n_j}{OPT_j}$ and the lemma follows.
\end{proof}

\smallskip

\begin{theorem} \label{theorem_weighted_unsplittable_logn}
Algorithm {\em Unsplit-Log-Approx} computes a $(\ln n)$-approximation for weighted capacitated domination problem with unsplittable demands in $O(n^3)$ time, where $n$ is the number of vertices.
\end{theorem}

\begin{proof}
To see that the algorithm produces a logarithmic approximation, take
the sum over each $S_j$, $1\le j \le m$ and observe that $n_{j+1} = n_j - k_j$, we have
$$\sum_{1\le j\le m}S_j \le \sum_{1\le j \le
m}\frac{k_j}{n_j}\cdot OPT_j \le \left(\sum_{1\le j \le n}\frac{1}{j}\right)\cdot OPT
\le \ln n \cdot OPT.$$
To see the time complexity, notice that it requires $O(n)$ time to compute a most efficient move for each vertex, which leads to an $O(n^2)$ computation for the most efficient choice in each iteration. The number of iterations is upper bounded by $O(n)$ since at least one vertex is satisfied in each iteration.
\end{proof}


\subsection{Weighted Splittable Demand} \label{section_weighted_splittable}

In this section, we
present an algorithm that produces a $(4\ln n+2)$-approximation for
the {\em weighted capacitated domination problem with splittable
demand}. The difference between this algorithm and the previous one lies in the way we handle the demand assignment.  In each iteration the demand of a vertex may be partially served. 
The unsatisfied portion of the demand is called {\it residue demand}. For each vertex $u \in V$, let $rd(u)$ be the residue demand of $u$. $rd(u)$ is set equal to $d(u)$ initially, and will be updated accordingly when a portion of the residue demand is assigned. $u$ is said to be completely satisfied when $rd(u) = 0$.

\smallskip

We will inherit the notation used in the previous
section. We assume that the elements of $N_{ud}[u]$,
written as $v_{u,1},v_{u,2}, \ldots ,v_{u,|N_{ud}[u]|}$, are sorted according to
their demands in non-decreasing order.

\smallskip

In each iteration, the algorithm performs two greedy choices. First, the
algorithm chooses the vertex of the most efficiency from $V$, where
the efficiency is defined similarly as in the previous section with some modification since the demand is splittable.

\begin{figure*}[t]
\rule{\linewidth}{0.2mm}
\medskip
{{\sc Algorithm} {\em Split-Log-Approx}}

\begin{algorithmic}[1]
\STATE $rd(u) \longleftarrow d(u)$, and $map(u) \longleftarrow \phi$ for each $u \in V$.
\WHILE{there exist vertices with non-zero residue demand}
    \STATE // {\bf $1^{st}$ greedy choice}
    \STATE Pick a vertex in $V$ with the most efficiency, say $u$.
    \IF{$j_u$ equals $0$}
        \STATE Assign this amount $c(u)\cdot \left\lfloor \frac{rd(v_{u,1})}{c(u)} \right\rfloor$ of residue demand of $v_{u,1}$ to $u$.
        \STATE $map(v_{u,1}) \longleftarrow \left\{u\right\}$
    \ELSE
        \STATE Assign the residue demands of the vertices in $\{v_{u,1}, v_{u,2}, \ldots, v_{u,j_u}\}$ to $u$.
        \IF{$j_u < \left|N_{ud}[u]\right|$}
            \STATE Assign this amount $c(u)-\sum_{i=1}^{j_u}rd(v_{u,i})$ of residue demand of $v_{u,j_u+1}$ to $u$.
            \STATE $map(v_{u,j_u+1}) \longleftarrow map(v_{u,j_u+1}) \cup \left\{u\right\}$
        \ENDIF
    \ENDIF
    \STATE
    \STATE // {\bf $2^{nd}$ greedy choice}
    \IF{there is a vertex $u$ with $0 < rd(u) < \frac{1}{2}\cdot d(u)$}
        \STATE Satisfy $u$ by doubling the demand assignment of $u$ to vertices in $map(u)$.
    \ENDIF
\ENDWHILE
\STATE compute from the assignment the cost of the dominating set, and return the result.

\end{algorithmic}
\rule{\linewidth}{0.2mm} \caption{The pseudo-code for the 
weighted splittable demand model.} \label{Algorithm for
Weighted Splittable Demand}
\end{figure*}

\smallskip

For each vertex $u \in V$, let $j_u$ with $0 \le j_u \le \left|N_{ud}[u]\right|$ be the maximum index such that
$c(u) \ge \sum_{i=1}^{j_u}rd(v_{u,i})$.
Let $X(u) = \sum_{i=1}^{j_u}\frac{rd(v_{u,i})}{d(v_{u,i})}$ be the sum of the effectiveness over the vertices whose residue demand could be completely served by a single copy of $u$. In addition, we let $$Y(u) = \frac{c(u)-\sum_{i=1}^{j_u}rd(v_{u,i})}{d(v_{u,j_u+1})}$$ if $j_u < \left|N_{ud}[u]\right|$ and $Y(u) = 0$ otherwise. The efficiency of $u$ is defined as $\frac{X(u)+Y(u)}{w(u)}$.

Second, the algorithm maintains for each vertex $u \in V$ a set of vertices, denoted by $map(u)$, which consists of vertices that have partially served the demand of $u$ before $u$ is completely satisfied. 
That is, for each $v \in map(u)$ we have a non-zero demand assignment of $u$ to $v$. Whenever there exists a vertex $u$ whose residue demand is below half of its original demand, i.e., $0 < rd(u) < \frac{1}{2}\cdot d(u)$, after the first greedy choice, the algorithm immediately doubles the demand assignment of $u$ to the vertices in $map(u)$. Note that in this way, we can completely satisfy the demand of $u$ since $\sum_{v \in map(u)} f(u,v) > \frac{1}{2}\cdot d(u)$ . A high-level description of this algorithm is presented in Figure~\ref{Algorithm for
Weighted Splittable Demand}.

\begin{observation} \label{observation_splittable_half}
After each iteration, the residue demand of each unsatisfied vertex is at least half of its original demand.
\end{observation}

Clearly, the observation holds in the beginning when the demand of each vertex is not yet assigned. For later stages, we argue that the algorithm properly maintains $map$ so that in our second greedy choice, whenever there exists a vertex $u$ for which $0 < rd(u) < \frac{1}{2}\cdot d(u)$, it's always sufficient to double the demand assignment $f(u,v)$ of $u$ to $v$ for each $v \in map(u)$. If $map(u)$ is only modified under the condition $0<j_v<\left|N_{ud}[v]\right|$, (line 12 in Figure~\ref{Algorithm for Weighted Splittable Demand}), then $map(u)$ contains exactly the set of vertices that have partially served $u$. As mentioned above, since $rd(u) < \frac{1}{2}\cdot d(u)$, it's sufficient to double the demand assignment in this case so $d(u)$ is completely satisfied. If $map(u)$ is reassigned through the condition $j_v = 0$ for some stage, then we have $c(v)<rd(u)\le d(u)$. Since we assign this amount $c(v)\cdot \left\lfloor rd(u)/c(v)\right\rfloor$ of residue demand of $u$ to $v$, this leaves at most half of the original residue demand and $u$ will be satisfied by doubling this assignment.

\smallskip

By the description given above, we conclude that the algorithm produces a feasible demand assignment as well as a feasible capacitated dominating set. Let the cost incurred by the first greedy choice be $S_1$ and the cost by the second choice be $S_2$. To see that the solution achieves the desired approximation guarantee, first notice that $S_2$ is bounded above by $S_1$, for what we do in the second choice is merely to satisfy the residue demand of a vertex, if there exists one, by doubling its previous demand assignment.

\smallskip

In the following, we will bound the cost $S_1$.
For each iteration $j$, let $u_j$ be the vertex of the maximum efficiency and $OPT_j$ be the cost of the optimal solution for the remaining problem instance. Let $n_j = \sum_{u \in V}rd(u)/d(u)$ denote the sum of effectiveness of each vertex in the remaining problem instance at the beginning of this iteration.
Let $S_{1,j}$ be the cost incurred by the first greedy choice in iteration $j$.
Assume that the algorithm repeats for $m$ iterations. We have the following lemma.

\begin{lemma} \label{lemma_splittable_greedy_choice}
For each $j$, $1 \le j \le m$, we have $S_{1,j} \le \frac{n_j - n_{j+1}}{n_j}\cdot OPT_j$, where $n_j - n_{j+1}$ is the effectiveness covered by $u_j$ in iteration $j$.
\end{lemma}

\begin{proof}
The optimality of our choice in each iteration is obvious since we assume that the elements of $N_{ud}[u]$ are sorted according to their original demands. Note that only in the case $c(u) < rd(v_{u,1})$, the algorithm could possibly take more than one copy. 
In this case the efficiency of our choice remains unchanged since the cost and the effectiveness covered by $u$ grows by the same factor. Therefore the efficiency of our choice, $(n_j - n_{j+1})/S_{1,j}$, is always no less than that of the optimal solution, which is $n_j / OPT_j$, and the lemma follows.
\end{proof}

\begin{observation} \label{observation_splittable_n_decrease}
We have $n_j - n_{j+1} \ge \frac{1}{2}$ for each $1 \le j \le m$.
\end{observation}

\begin{proof}
For iteration $j$, $1\le j\le m$, let $u$ be the vertex of the maximum efficiency. Observe that $v_{u,1}$ will be satisfied after this iteration. By Observation~\ref{observation_splittable_half}, we have $rd(v_{u,1}) / d(v_{u,1}) > \frac{1}{2}$. The lemma follows.
\end{proof}

\smallskip

By Lemma~\ref{lemma_splittable_greedy_choice} we have 
$$\sum_{j=1}^mS_{1,j} \le \sum_{j=1}^{m-1} \frac{n_j - n_{j+1}}{n_j}\cdot OPT_j+ \frac{n_m}{n_m}\cdot OPT_m \le \left(\sum_{j=1}^{m-1} \frac{\left\lceil n_j-n_{j+1}\right\rceil}{\left\lfloor n_j\right\rfloor} + 1\right)\cdot OPT,$$
since $\left\lfloor r\right\rfloor \le r \le \left\lceil r \right\rceil$ for any real number $r$ and $OPT_j \le OPT$ for each $1\le j\le m$.

\begin{lemma}
$\sum_{j=1}^{m-1}\left\lceil n_j-n_{j+1}\right\rceil / \left\lfloor n_j\right\rfloor \le 2\ln n$
\end{lemma}

\begin{proof}
Note that by Observation \ref{observation_splittable_half} and Observation \ref{observation_splittable_n_decrease}, we have $n_j > 1$ for all $j < m$. We will argue that this series together constitutes at most two harmonic series. By expanding the summand we have 
$$\frac{\left\lceil n_j-n_{j+1}\right\rceil}{\left\lfloor n_j\right\rfloor} \le \frac{1}{\left\lfloor n_j\right\rfloor}+\frac{1}{\left\lfloor n_j\right\rfloor-1}+\ldots+\frac{1}{\left\lfloor n_j\right\rfloor - \lceil n_j-n_{j+1}\rceil+1}.$$
Since $\lfloor n_{j+1} \rfloor = \lfloor n_j-(n_j-n_{j+1}) \rfloor \le \lfloor n_j \rfloor-\lfloor n_j-n_{j+1} \rfloor \le \lfloor n_j \rfloor
-\lceil n_j-n_{j+1}\rceil + 1$, the repetitions only occur at the first term and the last term if we expand the summation. By Observation \ref{observation_splittable_n_decrease}, the decrease of $n_j$ to $n_{j+1}$ is at least half. Therefore, the term $\left\lfloor n_j\right\rfloor - \left\lceil n_j-n_{j+1}\right\rceil + 1$ will never occur more than twice in the expansion. We conclude that
$\sum_{j=1}^{m-1}\left\lceil n_j-n_{j+1}\right\rceil / \left\lfloor n_j\right\rfloor \le 2 \ln n$
\end{proof}

\begin{theorem}
Algorithm {\em Split-Log-Approx}
computes a $(4\ln n + 2)$-approximation in $O(n^3)$ time,
where $n$ is the number of vertices, for
weighted capacitated domination problem with splittable demands.
\end{theorem}


\subsection{Unweighted Splittable Demand} \label{section_unweighted_splittable}

In this section, we consider the {\em unweighted capacitated domination problem with splittable demand} and present a $(2\ln n+1)$-approximation. In this case the weight $w(v)$ of each vertex $v \in V$ is considered to be 1 and the cost of the capacitated domination multiset $D$ corresponds to the total multiplicity of the vertices in $D$.   To this end, we first make a greedy reduction on the problem instance by spending at most $1\cdot OPT$ cost such that it takes at most one copy to satisfy each remaining unsatisfied vertex. Then we show that a $(2\ln n)$-approximation can be computed for the remaining problem instance, based on the same framework of Section~\ref{section_weighted_splittable}.

\smallskip

For each $u \in V$, let $g_u$ be the vertex in $N[u]$ with the maximum capacity. First, for each $u \in V$, we assign this amount $c(g_u) \cdot \left\lfloor \frac{d(u)}{c(g_u)} \right\rfloor$ of the demand of $u$ to $g_u$. Let the cost of this assignment be $S$, then we have the following lemma.

\begin{lemma} \label{lemma_last_copy}
We have $S \le OPT$, where $OPT$ is the cost of the optimal solution.
\end{lemma}

\begin{proof}
Notice that an optimal solution $O^*$ for the relaxation of this problem, where fractional copies are allowed, can be obtained by assigning the demand $d(u)$ of $u$ to $g_u$. Since $S \le O^*$ and $O^* \le OPT$, the lemma follows.
\end{proof}

In the following, we will assume that $d(u) \le c(g_u)$, for each $u \in V$. The algorithm of Section~\ref{section_weighted_splittable} is slightly modified. In particular, for the second greedy choice, whenever $rd(u) < d(u)$ for some vertex $u \in V$, we immediately assign the residue demand of $u$ to $g_u$. A high-level description of this algorithm 
is presented in Figure \ref{Algorithm for Unweighted Splittable Demand}.

\begin{figure*}[t]
\rule{\linewidth}{0.2mm}
\medskip
{{\sc Algorithm} {\em Unweighted-Split-Log-Approx}}

\begin{algorithmic}[1]
\STATE For each $u \in V$, assign $c(g_u) \cdot \left\lfloor \frac{d(u)}{c(g_u)} \right\rfloor$ demands of $u$ to $g_u$, where $g_u \in N[u]$ has the maximum capacity.
\STATE Reset the demands of the instance by setting $d(u) \longleftarrow rd(u)$ for each $u \in V$.
\WHILE{there exist vertices with non-zero residue demand}
    \STATE // {\bf $1^{st}$ greedy choice}
    \STATE Pick a vertex in $V$ with the most efficiency, say $u$.
        \STATE Assign the demands of the vertices in $\{v_{u,1}, v_{u,2}, \ldots, v_{u,j_u}\}$ to $u$.
        \IF{$j_u < \left|N_{ud}[u]\right|$}
            \STATE Assign this amount $c(u)-\sum_{i=1}^{j_u}rd(v_{u,i})$ of the residue demand of $v_{u,j_u+1}$ to $u$.
        \ENDIF
    \STATE
    \STATE // {\bf $2^{nd}$ greedy choice}
    \IF{there is a vertex $u$ with $0 < rd(u) < d(u)$}
        \STATE Satisfy $u$ by assigning the residue demand of u to $g_u$.
    \ENDIF
\ENDWHILE
\STATE compute from the assignment the cost of the dominating set, and return the result.

\end{algorithmic}
\rule{\linewidth}{0.2mm} \caption{The pseudo-code for the 
unweighted splittable demand model.} \label{Algorithm for Unweighted Splittable Demand}
\end{figure*}

\begin{observation} \label{observation_unweighted_splittable_n_decrese}
We have $n_j - n_{j+1} \ge 1$ for each $1 \le j \le m$.
\end{observation}

\begin{proof}
Observe that in each iteration, at least one vertex is satisfied and the residue demand of each unsatisfied vertex is equal to its original demand.
\end{proof}

Clearly, $S_2$ is bounded above by $S_1$, as we always take one copy for the first greedy choice and at most one copy for the second greedy choice in each iteration.
By Observation~\ref{observation_unweighted_splittable_n_decrese} and the fact that $n_j$ is integral for each $1 \le j \le m$, we have 
$$\sum_{j=1}^mS_{2,j} \le \sum_{j=1}^m \frac{n_j - n_{j+1}}{n_j}\cdot OPT_j \le \ln n\cdot OPT,$$
and $S+\sum_{j=1}^m\left(S_{1,j}+S_{2,j}\right) \le (2\ln n+1) \cdot OPT.$

We conclude the result as the following theorem.

\begin{theorem}
Algorithm {\em Unweight-Split-Log-Approx} computes
a $(2\ln n + 1)$-approximation in $O(n^3)$ time for weighted capacitated domination
problem with unsplittable demands, where $n$ is the number of
vertices.
\end{theorem}


\section{Parameterized Results} \label{parameterized_results}

\subsection{Hardness Results} \label{section_w_1_hardness}

In this section we show that {\em Capacitated Domination Problem} is {\it W[1]}-hard when parameterized by
treewidth by making a reduction from {\em k-Multicolor Clique}, a
restriction of {\em k-Clique} problem.

\begin{definition}[\sc Multicolor Clique]
Given an integer $k$ and a connected undirected graph $G = \left(\bigcup_{i=1}^k V[i], E\right)$ such that $V[i]$ induces an independent set for each $i$, the {\sc Multicolor Clique} problem asks whether or not there exists a clique of size $k$ in $G$.
\end{definition}

Given an instance $(G,k)$ of {\sc Multicolor Clique}, we will show how an instance $\G=(\V,\E)$ of {\em Capacitated Domination} with treewidth $O(k^2)$ can be built such that $G$ has a clique of size $k$ if and only if $\G$ has a capacitated dominating set of cost at most $k^\prime=(3k^2-k)/2$. For convenience, we shall distinguish the vertices of $\G$ by referring to them as {\em nodes}.

\smallskip

Let $N$ be the number of vertices. Without loss of generality, we label the vertices of $G$ by numbers, denoted $label(v)$, $v \in V$, between $1$ and $N$. For each $i\ne j$, let $E[i,j]$ denote the set of edges between $V[i]$ and $V[j]$. The graph $\G$ is defined as follows. For each $i$, $1\le i \le k$, we create a node $x_i$ with $w(x_i) = k^\prime+1$, $c(x_i) = 0$, and $d(x_i) = 1$. For each $u \in V[i]$, we have a node $\overline{u}$ with $w(\overline{u}) = 1$, $c(\overline{u}) = 1+(k-1)N$, and $d(\overline{u}) = 0$. We also connect $\overline{u}$ to $x_i$. For convenience, we refer to the star rooted at $x_i$ as vertex star $T_i$. 

\smallskip

Similarly, for each $1\le i<j\le k$, we create a node $y_{ij}$ with $w(y_{ij}) = k^\prime+1$, $c(x_i) = 0$, and $d(x_i) = 1$. For each $e\in E[i,j]$ we have a node $\overline{e}$ with $w(\overline{e}) = 1$, $c(\overline{e}) = 1+2N$, and $d(\overline{e}) = 0$. We connect $\overline{e}$ to $y_{ij}$. We refer to the star rooted at $y_{ij}$ as edge star $T_{ij}$. The selection of nodes in $T_i$ and $T_{ij}$ in the capacitated dominating set will correspond to the choices made in selecting the vertices that form a clique in $G$.

\smallskip

In addition, for each $i\ne j$, $1 \le i,j \le k$, we create two bridge nodes $b^1_{i,j}$, $b^2_{i,j}$ with $w(b^1_{i,j}) = w(b^2_{i,j}) = 1$ and $d(b^1_{i,j}) = d(b^2_{i,j}) = 1$. The capacities of the bridge nodes are to be defined later. Now we describe the way how stars $T_i$ and $T_{ij}$ are connected to bridge nodes such that the reduction claimed above holds. For each $i\ne j$, $1 \le i,j \le k$ and for each $v \in V[i]$, 
we create two propagation nodes $p^1_{v,i,j}$, $p^2_{v,i,j}$ and connect them to $\overline{v}$. Besides, we connect $p^1_{v,i,j}$ to $b^1_{i,j}$ and $p^2_{v,i,j}$ to $b^2_{i,j}$. We set $w(p^1_{v,i,j}) = w(p^2_{v,i,j}) = k^\prime+1$ and $c(p^1_{v,i,j}) = c(p^2_{v,i,j}) = 0$. The demands of $p^1_{v,i,j}$ and $p^2_{v,i,j}$ are set to be $d(p^1_{v,i,j}) = label(v)$ and $d(p^2_{v,i,j}) = N - label(v)$. For each $1 \le i<j\le k$ and for each $e = (u,v) \in E[i,j]$, we create four propagation nodes $p^1_{e,i,j}$, $p^2_{e,i,j}$, $p^1_{e,j,i}$, and $p^2_{e,j,i}$ with zero capacity and $k^\prime+1$ cost. Without loss of generality, we assume that $u \in V[i]$ and $v \in V[j]$. The demands of the four nodes are set as the following: $d(p^1_{e,i,j}) = N-label(u)$, $d(p^2_{e,i,j}) = label(u)$, $d(p^1_{e,j,i}) = N-label(v)$, and $d(p^2_{e,j,i}) = label(v)$. 
Finally, for each bridge node $b$, we set $c(b) = \sum_{u \in N[b]}d(b) - N$.

\begin{figure}[h]
\centering
\caption{The connections between stars and bridge nodes.}
\includegraphics{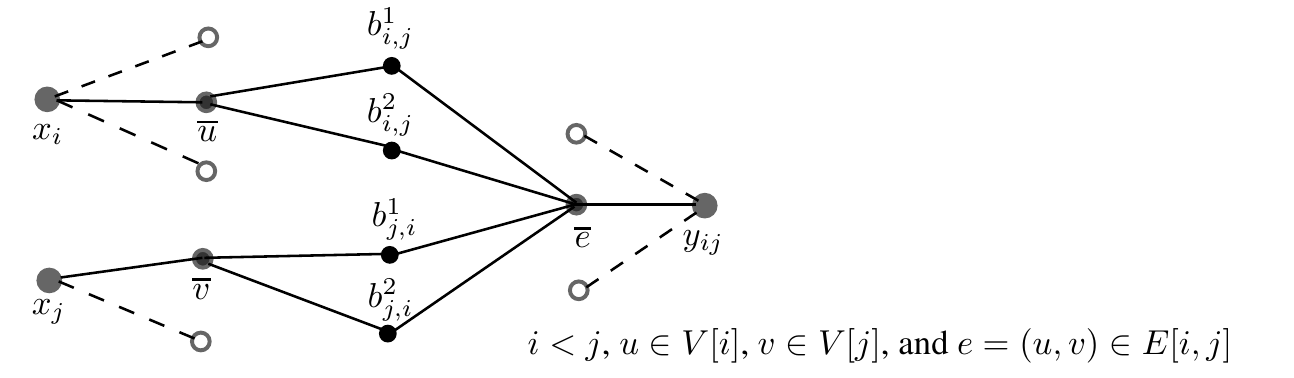}
\end{figure}

\begin{lemma} \label{lemma_w_1_hardness}
The treewidth of $\G$ is $O(k^2)$. Furthermore, $G$ admits a clique of size $k$ if and only if $\G$ admits a capacitated dominating set of cost at most $k^\prime = (3k^2-k)/2$.
\end{lemma}

\begin{proof}[Proof of Lemma~\ref{lemma_w_1_hardness}.]
Consider the set of bridge nodes, $Bridge = \bigcup_{i \ne j}\left\{b^1_{i,j} \cup b^2_{i,j}\right\}$. Since $\G \backslash Bridge$ is a forest and the removal of a vertex from a graph decreases the treewidth of the graph by at most one, the treewidth of $\G$ is upper bounded by the number of bridge nodes plus some constant, which is $O(k^2)$.

Let $C$ be a clique of size $k$ in $G$. By choosing the bridge nodes, $b^1_{i,j}$ and $b^2_{i,j}$ for each $i\ne j$, $\overline{u}$ for each $u \in C$, and $\overline{e}$ for each $e \in C$ exactly once, we have a vertex subset of cost exactly $(3k^2-k)/2$. One can easily verify that this is also a feasible capacitated dominating multi-set for $\G$.

On the other hand, let $D$ be a capacitated donimating multi-set of cost at most $k^\prime$ in $\G$. First observe that none of the propagation nodes are chosen in $D$, otherwise the cost would exceed $k^\prime$. This implies $b^1_{i,j} \in D$ and $b^2_{i,j} \in D$ for each $i\ne j$. Note that this already contributes cost at least $k(k-1)$ to $D$ and the rest of the nodes in $D$ together contributes at most $k(k+1)/2$. 

Similarly, we have $x_i \notin D$ and $y_{ij} \notin D$ for each $i \ne j$. Therefore, for each $1\le i\le k$, $\exists u \in V[i]$ such that $\overline{u} \in D$, and for each $i \ne j$, $\exists e \in E[i,j]$ such that $\overline{e} \in D$. Since we have $k(k-1)/2+k = k(k+1)/2$ such stars, exactly one node from each star is chosen to be included in $D$ and each node of $D$ is chosen exactly once. Next we argue that the nodes chosen in each star will correspond to a clique of size $k$ in $G$.

For each $1\le i<j \le k$, let $\overline{u} \in T_i$ and $\overline{v} \in T_j$ be the nodes chosen in $D$. Let $\overline{e} \in T_{ij}$ be the node chosen in $D$. In the following we shall prove that $e = (u,v)$. Since the capacity of $\overline{u}$ equals the sum of the demands over $N[\overline{u}]$, the closed neighborhood of $\overline{u}$, without loss of generality we can assume that the demands of nodes in $N[\overline{u}]$ are served by $\overline{u}$. Consider the bridge vertex $b^1_{ij}$ and the set $S = N[b^1_{ij}] \backslash N[\overline{u}]$. The demands of vertices in $S$ can only be served by either $b^1_{ij}$ or $\overline{e}$, as they are the only two vertices in $N[S]$ chosen to be included in $D$. In particular, vertices in $S$ apart from $p^1_{e,i,j}$ can only be served by $b^1_{ij}$. Notice that the sum of the demands in $S$ is $N-label(u)$ above the capacity of $b^1_{ij}$. Therefore we have $d(p^1_{e,i,j}) \ge N-label(u)$, which implies $d(p^2_{e,i,j}) \le label(u)$ as well since we have $d(p^1_{e,i,j})+d(p^2_{e,i,j})=N$ by our setting.

By a symmetric argument on $b^2_{i,j}$ we obtain $d(p^2_{e,i,j}) \ge label(u)$. Hence $d(p^2_{e,i,j}) = label(u)$. By another symmetric argument on $b^1_{ji}$ and $b^2_{ji}$, we have $d(p^2_{e,j,i}) = label(v)$. Therefore $e = (u,v)$ by our construction.
\end{proof}

Note that this proof holds for both splittable and unsplittable demand models. We have the following theorem.

\begin{theorem} \label{theorem_w_1_hardness}
The {\em Capacitated Domination problem} is {\it W[1]}-hard when parameterized by treewidth.
\end{theorem}

\begin{proof}[Proof of Theorem~\ref{theorem_w_1_hardness}.]
This theorem follows directly from Lemma~\ref{lemma_w_1_hardness} and the fact that this reduction can be computed in time polynomial in both $k$ and $N$.
\end{proof}


\subsection{FPT Algorithms on Graphs of Bounded Treewidth}

In this section we show that {\em Capacitated Domination Problem} with unsplittable demand is FTP when parameterized by both treewidth and maximum capacity by giving a $2^{2k(\log M+1)+\log k+O(1)}\cdot n$ exact algorithm.

\smallskip

To this end, we give a dynamic programming algorithm on a so-called {\em nice tree decomposition} \cite{DBLP:books/sp/Kloks94} of the input graph G. In the following, without loss of generality, we shall assume that the bag associated with the root of $T$ is empty. For each node $i$ in the tree $T$, let $T_i$ be the subtree rooted at $i$ and $Y_i := \bigcup_{j \in T_i}X_j$. Starting from the leaf nodes of $T$, our algorithm proceeds in a bottom-up manner and maintains for each node $i$ of $T$ a table $A_i$ whose columns consist of the following information.

\begin{itemize}
\item{$P$ with $P \subseteq X_i$ indicating the set of vertices in $X_i$ that have been served, and}
\item{$rc(u)$ with $0 \le rc(u) < c(u)$ indicating the residue capacity of $u$, for each $u \in X_i$.}
\end{itemize}

Clearly, each row of $A_i$ corresponds to a possible configuration consisting of the unsatisfied vertices and the residue capacity of each vertex in $X_i$ that can be used. The algorithm computes for each row of $A_i$ the cost of the optimal solution to the subgraph induced by $Y_i$ under the constraint that the configuration of vertices in $X_i$ agrees with that specified by the values of the row.

\smallskip

In the following, we describe the computation of the table $A_i$ for each node $i$ in the tree $T$ in more detail. In order to keep the content clean, we use the terms "insert a new row" and "replace an old row by the new one" interchangeably. Whenever the algorithm attempts to insert a new row into a table while another row with identical configuration already exists, the one with the smaller cost will be kept.  According to different types of vertices we encounter during processing, we have the following situations.

\begin{itemize}
\item{\bf $i$ is a leaf node.} Let $X_i = \left\{v\right\}$. We add two rows to the table $A_i$ which correspond to cases whether or not $v$ is served.

\smallskip

\begin{algorithmic}[1]
\STATE let $r_1 = (\left\{ \phi \right\}, \left\{rc(v)=0\right\})$ be a new row with $cost(r_1) \longleftarrow 0$
	\STATE let $r_2 = (\left\{v\right\}, \{rc(v)=d(v)$ mod $c(v)\})$ be a new row with $cost(r_2) \longleftarrow w(v) \cdot \left\lceil \frac{d(v)}{c(v)} \right\rceil$
	\STATE add $r_1$ and $r_2$ to $A_i$
	\newline \rule{\linewidth}{0.2mm}
\end{algorithmic}

\item{\bf $i$ is an introduce node.} Let $j$ be the child of $i$, and let $X_i = X_j \cup \left\{v\right\}$. The data in $A_j$ is basically inherited by $A_i$.We extend $A_i$ by considering, for each existing row $r$ in $A_j$, all $2^{\left|X_j \backslash P_r\right|}$ possible ways of choosing vertices in $X_j \backslash P_r$ to be assigned to $v$. In addition, $v$ can be either unassigned or assigned to any vertex in $X_i$. In either case, the cost and the residue capacity are modified accordingly. 

\smallskip

\begin{algorithmic}[1]
\FORALL{row $r_0 = (P, R) \in A_j$}
    \FORALL{possible $U$ such that $U \subseteq (X_j \backslash P) \cap N_G(v)$}
        \STATE let $R^\prime = R \cup \{ rc(v) = \sum_{u \in U}d(u)$ mod $c(v) \}$, and
        \newline let $r = (P \cup U, R^\prime)$ be a new row with
        \newline $cost(r) = cost(r_0) + w(v) \cdot \left\lceil \frac{\sum_{u \in U}d(u)}{c(v)} \right\rceil$
        \STATE add $r$ to $A_i$
        \FORALL{$u \in X_i$}
            \STATE let $r^\prime = (P \cup U \cup \left\{v\right\}, R^\prime \cup \{rc(u) = (rc(u) - d(v))$ mod $c(u)\})$ be a new row with
            \newline $cost(r^\prime) = cost(r) +$ the cost required by this assignment
            \STATE add $r^\prime$ to $A_i$
        \ENDFOR
    \ENDFOR
\ENDFOR
\newline \rule{\linewidth}{0.2mm}
\end{algorithmic}

\item{\bf $i$ is a forget node.} Let $j$ be the child of $i$, and let $X_i = X_j \backslash \left\{v\right\}$. In this case, for each row $r \in A_{j}$ such that $v\in P_r$, we insert a row $r^\prime$ to $A_i$ identical to $r$ except for the absence of $v$ in $P_{r^\prime}$. The remaining rows in $A_j$, which correspond to situations where $v$ is not served, are ignored without being considered.

\smallskip

\begin{algorithmic}[1]
\FORALL{row $r_0 = (P, R) \in A_j$ such that $v \in P$}
	\STATE let $r = (P \backslash \left\{v\right\}, R \backslash \left\{rc(v)\right\})$ be a new row with $cost(r) = cost(r_0)$
	\STATE add $r$ to $A_i$
\ENDFOR
\newline \rule{\linewidth}{0.2mm}
\end{algorithmic}

\item{\bf $i$ is a join node.} Let $j_1$ and $j_2$ be the two children of $i$ in $T$. We consider every pair of rows $r_1$, $r_2$ where $r_1 \in A_{j_1}$ and $r_2 \in A_{j_2}$. We say that two rows $r_1$ and $r_2$ are {\em compatible} if $P_{r_1} \cap P_{r_2} = \phi$. For each compatible pair of rows $\left(r_1,r_2\right)$, we insert a new row $r$ to $A_i$ with $P_r = P_{r_1} \cup P_{r_2}$, $rc_r(u) = \left(rc_{r_1}(u)+rc_{r_2}(u)\right)$ mod $c(u)$, for each $u \in X_i$, and $cost(r) = cost(r_1)+cost(r_2) - \sum_{u \in X_i}\left\lfloor\frac{rc_{r_1}(u)+rc_{r_2}(u)}{c(u)}\right\rfloor$.

\smallskip

\begin{algorithmic}[1]
\FORALL{compatible pairs $r_1 = (P_1, R_1) \in A_{j_1}$ and $r_2 = (P_2, R_2) \in A_{j_2}$}
	\STATE let $r = \left(P_1 \cup P_2, R \right)$ be a new row.
	\STATE $cost(r) \longleftarrow cost(r_1) + cost(r_2) - \sum_{u \in X_i}\left\lfloor\frac{rc_{R_1}(u) + rc_{R_2}(u)}{c(u)}\right\rfloor$, and
	\STATE $R \longleftarrow \{ rc_{R_1}(u) + rc_{R_2}(u)$ mod $c(u) : u \in X_i\}$
	\STATE add $r$ to $A_i$
\ENDFOR
\newline \rule{\linewidth}{0.2mm}
\end{algorithmic}

\end{itemize}

\begin{theorem}
Capacitated Domination problem with unsplittable demand on graphs of bounded treewidth can be solved in time $2^{2k(\log M +1)+\log k+O(1)}\cdot n$, where $k$ is the treewidth and $M$ is the largest capacity.
\end{theorem}

\begin{proof}
The correctness of the algorithm follows from the description above. The running time for computing the table $A_i$ associated with each tree node $i$ is bounded above by the time taken on the join nodes, which is clearly $2^{2k(\log M +1)+\log k}\cdot n$. The theorem follows.
\end{proof}

\smallskip

We state without going into details that by suitably replacing the set $P_i$ with the residue demand $rd_i(u)$ for each vertex $u \in X_i$ in the column of the table we maintained, the algorithm can be modified to handle the splittable demand model. We have the following corollary.

\begin{corollary}
Capacitated Domination problem with splittable demand on graphs of bounded treewidth can be solved in time $2^{(2M+2N+1)\log k+O(1)}\cdot n$, where $k$ is the treewidth, $M$ is the largest capacity, and $N$ is the largest demand.
\end{corollary}


\subsection{Extension to Planar Graphs}

In this section we extend the above FPT algorithms based on a framework due to Baker \cite{174650} to obtain a pseudo-polynomial time approximation scheme for planar graphs. In particular, for unsplittable demand model, given a planar graph $G$ with maximum capacity $M$ and an integer $k$, the algorithm computes an $(1+\frac{4}{k-1})$-approximation in time $O(2^{2k(\log M+1)+2\log k}n)$, where $n$ is the number of vertices. Taking $k = \left\lceil c\log n\right\rceil$, where $c$ is some constant, we get a pseudo-polynomial time approximation algorithm which converges toward optimal as $n$ increases. On the other hand, for splittable demand model, we have a pseudo-polynomial time approximation scheme in $O(2^{(2M+2N+1)\log k+O(1)}\cdot n)$ time, where $N$ is the maximum demand. To get rid of the factor $N$, we could apply the transformation used in Section~\ref{section_unweighted_splittable} and Lemma~\ref{lemma_last_copy} in advance and obtain a $(2+\frac{4}{k-1})$-approximation in $O(2^{(4M+1)\log k+O(1)}\cdot n)$ time.

\smallskip

This is done as follows. Given a planar graph $G$, we generate a planar embedding and retrieve the vertices of each level using the linear-time algorithm of Hopcroft and Tarjan \cite{321852}. Let $m$ be the number of levels of this embedding. Let $OPT$ be the cost of the optimal capacitated dominating set of $G$, and $OPT_j$ be the cost contributed by vertices at level $j$. Since $\sum_{0\le i\le m}(OPT_i+OPT_{i+1})\le 2\cdot OPT$, there exists one $r$ with $0\le r<k$ such that 
$$\sum_{0\le j<\left\lfloor \frac{m}{k}\right\rfloor}\left(OPT_{jk+r}+OPT_{jk+r+1}\right)\le \frac{2}{k}\cdot OPT.$$
For each $0\le j \le \left\lfloor \frac{m}{k}\right\rfloor +1$, let $G_j$ be the graph induced by vertices between level $(j-1)k+r$ and $jk+r+1$. In addition, we set the demands of vertices at level $(j-1)k+r$ and level $jk+r+1$ to be zero for each $G_j$. Clearly, the treewidth of each $G_j$ is upper bounded by $k+1$ and the sum of the optimal cost for each $G_j$ is no more than $(1+\frac{4}{k})\cdot OPT$. Take $k^\prime = k-1$ and we're done.


\section{Concluding Remarks} \label{conclusion}

In this paper we considered the {\em Capacitated Domination} problem, which is a generalization of the well-known {\em Dominating Set} problem and which models a service-requirement assignment scenario. In terms of polynomial time approximations, we have presented logarithmic approximation algorithms with respect to different demand assignment models for this problem on general graphs. Together with our previous
work on generally approximating this problem, this establishes the corresponding approximation results to the well-known approximations of the traditional {\em Dominating Set} problem and closes the problem of generally approximating the optimal solution. On the other hand, from the perspective of parameterization, we have proved that this problem is {\it W[1]}-hard when parameterized by treewidth of the graph. Based on this hardness result, we presented exact FPT algorithms when parameterized by treewidth and maximum capacity of the vertices. This algorithm is further extended under a framework of Baker \cite{174650} to obtain approximations for planar graphs.

We conclude with a few open problems and future research goals. First, although exact FPT algorithms are provided, the problem of approximating the optimal solution when parameterized by treewidth remains open. It would be nice to obtain faster approximation algorithms for graphs of bounded treewidth as this would provide faster approximations for planar graphs as well. Second, it would be nice to know how the problem behaves on special graph classes. As this problem has been shown to be difficult and admit a PTAS on trees when the demand can be split, approximations for other classes such as interval graphs remain unknown. Third, from the perspective of parameterization, it may be possible to find other parameters that are more closely related to the problem and obtain better results.


\bibliographystyle{plain}

\small
\bibliography{approx_capacitated_domination}

\end{document}